\documentclass[submission,copyright,creativecommons,fleqn]{eptcs}
\providecommand{\event}{GandALF 2012} 

\usepackage{amsmath}
\usepackage{amsfonts}
\usepackage{amsthm}
\usepackage{amssymb}
\usepackage{url}
\usepackage{fancybox}
\usepackage{tikz}
\usepackage{enumerate} 

\usetikzlibrary{decorations.pathreplacing}
\usetikzlibrary{decorations.pathmorphing}
\usetikzlibrary{shapes,shapes.symbols,automata,arrows,shadows}
\usetikzlibrary{calc}

\tikzset{p0/.style = {shape = circle,    draw, thick, minimum size = 0.4cm}}
\tikzset{p1/.style = {shape = rectangle, draw, thick, minimum size = 0.4cm}}
\tikzset{>=stealth, shorten >=1pt}
\tikzset{every edge/.style = {thick, ->, draw}}
\tikzset{every loop/.style = {thick, ->, draw}}

\renewcommand{\epsilon}{\varepsilon}
\newcommand{\Powerset}{\mathcal{P}}
\newcommand{\Parity}{\mathrm{Par}}
\newcommand{\WaluCheck}{\mathsf{Check}}
\newcommand{\WaluClaim}{\mathsf{Claim}}
\newcommand{\WaluJump}{\mathsf{Jump}}
\newcommand{\WaluWin}{\mathsf{Win}}
\newcommand{\WaluPush}{\mathsf{Push}}

\newcommand{\Push}{\mathrm{push}}
\newcommand{\Pop}{\mathrm{pop}}
\newcommand{\Skip}{\mathrm{skip}}

\newcommand{\MinCol}{\mathsf{MinCol}}
\newcommand{\Sc}{\mathsf{Sc}}
\newcommand{\StairSc}{\mathsf{StairSc}}
\newcommand{\StairPositions}{\mathsf{StairPositions}}
\newcommand{\Stairs}{\mathsf{Stairs}}
\renewcommand{\Reset}{\mathsf{reset}}
\newcommand{\lastBump}{\mathsf{lastBump}}

\newcommand{\lastStrictBump}{\mathsf{lastStrictBump}}
\newcommand{\last}{\mathsf{last}}

\newcommand{\Min}{\mathsf{min}}
\newcommand{\Max}{\mathsf{max}}
\newcommand{\MaxSc}{\mathsf{MaxSc}}
\newcommand{\MaxStairSc}{\mathsf{MaxStairSc}}

\newcommand{\nats}{\mathbb{N}}
\newcommand{\Inf}{\mathrm{Inf}}

\newcommand{\sh}{\mathrm{sh}}
\newcommand{\col}{\mathrm{col}}

\newcommand{\strat}{\sigma}
\newcommand{\Prediction}{\mathsf{Pred}}
\newcommand{\Gammabot}{\Gamma_{\!\!\bot}}

\theoremstyle{plain}
\newtheorem{example}{Example}
\newtheorem{theorem}{Theorem}
\newtheorem{lemma}[theorem]{Lemma}
\newtheorem{definition}[theorem]{Definition}
\newtheorem{remark}[theorem]{Remark}

\title{Playing Pushdown Parity Games in a Hurry}
\author{Wladimir Fridman
\institute{Chair of Computer Science 7 \\
RWTH Aachen University\\ Aachen, Germany}
\email{fridman@automata.rwth-aachen.de}
\and
Martin Zimmermann
\institute{Institute of Informatics\\
University of Warsaw\\
Warsaw, Poland}
\email{zimmermann@mimuw.edu.pl}
}
\def\titlerunning{Playing Pushdown Parity Games in a Hurry}
\def\authorrunning{W. Fridman \& M. Zimmermann}

\begin{document}
\maketitle

 \begin{abstract}
We continue the investigation of finite-duration variants of infinite-duration
games by extending known results for games played on finite graphs to those
played on infinite ones. In particular, we establish an equivalence between
pushdown parity games and a finite-duration variant. This allows us to
determine the winner of a pushdown parity game by solving a reachability game on a
finite tree.
 \end{abstract}

\section{Introduction}

Infinite two-player games on graphs are a powerful tool to model, verify, and
synthesize open reactive systems and are closely related to fixed-point logics.
The winner of a play in such a game typically emerges only after completing the
whole (infinite) play. Despite this, McNaughton became interested in playing
infinite games in finite time, motivated by his belief that ``infinite games
might have an interest for casual living room
recreation''~\cite{DBLP:conf/hcat/McNaughton00}.

As playing infinitely long is impossible for human players, McNaughton
introduced scoring functions for Muller games, a certain type of infinite game.
Each of these functions is associated to one of the two players, so it makes
sense to talk about the scores of a player. The scoring functions are updated
after every move and describe the progress a player has made towards winning the
play.  However, as soon as a scoring function reaches its predefined threshold,
the game is stopped and the player whose score reached its threshold first is
declared to win this (now finite) play.

On the theoretical side, by applying finite-state determinacy of Muller games,
McNaughton showed that a Muller game and a finite-duration variant with a
factorial threshold score have the same winner. Thus, the winner of a Muller
game can be determined by solving a finite reachability game, which is much
simpler to solve, albeit doubly-exponentially larger than the original  Muller
game. 

This result was improved by showing that the finite-duration game with
threshold three always has the same winner as the original Muller
game~\cite{DBLP:journals/corr/abs-1006-1410} and by a (score-based) reduction
from a Muller game to a safety game whose solution not only yields the winner of
the Muller game, but also a winning strategy~\cite{NRZ12}. The
improved threshold does not rely on finite-state determinacy, but is proven by
constructing strategies that are winning for both games at the same time.

The reduction from Muller to safety games yields a new memory structure for
Muller games that implements not only a winning strategy, but also the most
general non-deterministic winning strategy (a so-called permissive strategy)
that prevents the losing player from reaching a certain score. This extends the
work of Bernet et al.\ on permissive strategies for parity
games~\cite{DBLP:journals/ita/BernetJW02} to Muller games. For
parity games, the algorithm presented to compute a permissive strategy is
Jurdzi\'{n}ski's progress measure algorithm~\cite{DBLP:conf/stacs/Jurdzinski00}
for solving parity games. This raises the question of whether there is also a
(score-based) progress measure algorithm for Muller games, which can be derived
from the construction of a permissive~strategy.

In this work, we begin to extend these results to infinite games played on 
infinite game graphs. At first, two questions have to be answered: what type 
of infinite game graphs and what type of winning condition to consider? We have 
to restrict the type of game graphs, since stopping a play after a finite number of
rounds can only lead to an equivalent finite-duration variant if there is some 
regularity in the game graph. A well-researched class of infinite graphs are 
configuration graphs of pushdown systems. Walukiewicz showed how to solve
parity games on such game graphs in exponential time by a reduction to parity 
games on finite game graphs~\cite{DBLP:conf/cav/Walukiewicz96}. 

As for the second question, we also consider parity games. In his work on making
infinite games playable for human players, McNaughton was interested  in Muller
games, since he thought that games for human players should not be positionally
determined in order to be interesting enough. From a theoretical point of view,
this can be argued as follows: every positionally determined game in a finite
game graph has a trivial finite-duration variant. In this variant, a play is
stopped as soon as a vertex is visited for the second time and the winner is the
player who wins the infinite play induced by this cycle. As every positional
winning strategy for the infinite-duration game is also winning for the
finite-duration game, the two games have the same winner.

For a (min-) parity game on a finite game graph this criterion can be improved:
let $|V|_c$ denote the number of vertices colored by $c$. Then, a positional
winning strategy for Player~$i\in \{0, 1\}$ does not visit $|V|_c+1$ vertices of
color $c$ with parity $1-i$ without visiting a vertex of smaller color in between. This
condition can be expressed using scoring functions $\Sc_c$ 
that count the number of vertices of color~$c$ visited since the last visit of a
vertex of color~ $c'<c$. Due to positional determinacy of parity games, the
following finite-duration game has the same winner as the original parity game:
a play is stopped as soon as some scoring function $\Sc_c$ reaches value $|V|_c+1$
for the first time and Player~$i$ is declared to be the winner, if the parity
of $c$ is $i$. Again, a positional winning strategy for the parity game
is also winning for the finite-duration game, i.e., the two games indeed have the same
winner.

However, both criteria do not necessarily yield a finite-duration game when
applied to a game on an infinite game graph; the first one since there could be
infinite simple paths, the second one since there are colors that color
infinitely many vertices, i.e., $|V|_c$ could be infinite. Hence, devising a
finite-duration variant of games on infinite game graphs requires more
sophisticated criteria, even if the game is positionally determined.

We exploit the intrinsic structure of the game graph induced by the pushdown 
system by defining stair-score functions
$\StairSc_c$ for every color~$c$ and show the equivalence between a parity game
and the finite-duration version, when played up to an exponential threshold
stair-score (in the size of the pushdown system). This result shows how to
determine the winner of an infinite game on an infinite game graph by solving a
finite reachability game. We complement this by giving a lower bound on the
threshold stair-score that always yields the same winner, which is exponential
in the cubic root of the size of the underlying pushdown system.

To prove our main theorem, we analyze Walukiewicz's reduction from parity games
on pushdown graphs to parity games on finite graphs and prove a correspondence
between stair-scores in the pushdown game and scores in the finite parity game. 
The winning player of the finite parity game (who
also wins the pushdown game) has a winning strategy that bounds the losing
player's scores by $|V|_c$ (the number of vertices colored by $c$ in the finite
parity game). We show that this strategy can be turned into a winning strategy
for him in the pushdown game that bounds the stair-scores by $|V|_c$ as well.
Since the finite parity game is of exponential size, our result follows.  

This work is organized as follows: after fixing our notation
for parity games and pushdown systems in Section~\ref{sec_defs}, we introduce
the score and stair-score functions in Section~\ref{sec_scores}. In
Section~\ref{sec_walu}, we recall Walukiewicz's reduction, which we apply in
Section~\ref{sec_main} to prove our main theorem, namely the equivalence between
parity games on pushdown graphs and their finite-duration variant. Finally, in
Section~\ref{sec_lowerbounds}, we prove the lower bounds on the threshold score
that always yields an equivalent finite-duration game.

\section{Preliminaries}
\label{sec_defs}

The power set of a set $X$ is denoted by $\Powerset(X)$. 
The set of non-negative integers is denoted by $\nats$. For $n\in\nats$,
let $[n]=\{0,\ldots,n-1\}$ and $\Parity(n)=0$ if $n$ is even, and 
$\Parity(n)=1$ if $n$ is odd. Moreover, for every alphabet $\Sigma$, i.e., a finite set of symbols, 
the set of finite words is denoted by $\Sigma^*$, 
and $\Sigma^\omega$ denotes the set of infinite words. 
The length of a word $w\in\Sigma^*$ is denoted by $|w|$ and $\epsilon$ denotes the empty word,
i.e., the word of length $|\epsilon|=0$.
For $n\in\nats$, the set of words of length at most $n$ is denoted by 
$\Sigma^{\leq n}$ and for $\Sigma^*\setminus\{\epsilon\}$ we also write $\Sigma^+$.
For a word $w \in \Sigma^+ \cup \Sigma^\omega$ 
and $n\in\nats$, 
we write $w(n-1)$ for the $n$-th letter of $w$ (the first letter is $w(0)$) and denote its last letter by $\last(w)$. For $w \in \Sigma^*$ and $w' \in \Sigma^* \cup \Sigma^\omega$, we write
$w \sqsubseteq w'$ if $w$ is a prefix of $w'$ and $w \sqsubset w'$ if $w$ is a strict prefix of $w'$. For a word $\rho\in \Sigma^\omega$, let 
$\Inf(\rho)=\{a \in \Sigma \mid \rho(n)=a \text{ for infinitely many } n\}$. 

\subsection{Parity Games} A game graph is a tuple 
$G=(V,V_0,V_1,E, v_{\mathrm{in}})$ where $(V,E)$ is a (possibly countably infinite) 
directed graph with set $V$ of vertices and set $E\subseteq V\times V$ 
of edges, where $V_0\cup V_1$ is a partition of $V$ and $v_{\mathrm{in}}\in V$ 
is the initial vertex. We assume that every vertex has at least one outgoing 
edge. Vertices from $V_i$ belong to Player~$i$, for $i\in\{0,1\}$.

A parity game $\mathcal G= (G, \col)$ consists of a game graph $G$ and a
coloring function $\col\colon V \rightarrow [n]$, for some $n\in\nats$. Given
$\col$, we define $\MinCol\colon V^+ \rightarrow [n]$ by
$\MinCol(w)=\Min\{\col(w(i)) \mid 0\le i< |w| ) \}$. A play of $\mathcal G$ is
built up by the two players by moving a token on the game graph. Initially, the
token is placed on $v_{\mathrm{in}}$. In every round, if the current vertex $v$
is in $V_i$, then Player $i$ has to choose an outgoing edge $(v,v')\in E$ and
the token is moved to the successor $v'$. Thus, a play in $\mathcal G$ is an
infinite sequence $\rho \in V^\omega$ such that $\rho(0)= v_{\mathrm{in}}$ and
$(\rho(n), \rho(n+1))\in E$ for every $n \in \nats$. Such a play $\rho$ is
winning for Player~$0$ if $\Min\{\Inf(\col(\rho))\}$ is even, otherwise it is
winning for Player~$1$. Here, $\col(\rho)$ represents the sequence of colors
seen by $\rho$. Thus, we sometimes refer to the coloring function $\col$ as a
min-parity condition.

A strategy for Player $i$ is a function $\strat\colon V^*V_i \rightarrow V$ 
such that $(\last(w),\strat(w))\in E$ for every $w \in V^* V_i$. A strategy 
$\strat$ is called positional if $\strat(w)=\strat(w')$ holds for all 
$w,w'\in V^*V_i$ with $\last(w)=\last(w')$. 
A play $\rho$ is consistent with $\strat$ for Player $i$ if 
$\rho(n+1)=\strat(\rho(0)\cdots\rho(n))$ for every $n\in\nats$ with 
$\rho(n)\in V_i$. A strategy $\strat$ is a winning strategy for 
Player $i$ if every play~$\rho$ that is consistent with $\strat$ 
is winning for Player~$i$. We say that Player~$i$ wins~$\mathcal G$
if there exists  a winning strategy for Player~$i$. A game is determined if one of the players wins it.

\begin{theorem}[\cite{DBLP:conf/focs/EmersonJ91,Mos91}]
 Parity games are determined with positional winning strategies. 
\end{theorem}

\subsection{Pushdown Game Graphs} 
A pushdown system (PDS) $\mathcal P = (Q,\Gamma,\Delta, q_{\mathrm{in}})$ 
consists of a finite set of states $Q$ with an initial state 
$q_{\mathrm{in}}\in Q$, a stack alphabet $\Gamma$ with the initial 
stack symbol $\bot\notin\Gamma$, which can neither be written nor 
deleted from the stack, and a transition relation 
$\Delta\subseteq Q \times \Gammabot \times Q \times \Gammabot^{\leq 2}$, 
where $\Gammabot=\Gamma\cup\{\bot\}$. We say that  a transition 
$\delta=(q,A,q',\alpha)\in\Delta$ is a 
$\Push$-transition if $|\alpha|=2$, $\delta$ is a 
$\Skip$-transition if $|\alpha|=1$, and $\delta$ is a 
$\Pop$-transition if $\alpha=\epsilon$. 
In the following, we assume
every PDS to be deadlock-free, i.e., for every $q\in Q$ and $A\in\Gammabot$
there exist $q'\in Q$ and $\alpha\in\Gammabot^{\leq 2}$ 
such that $(q,A,q',\alpha)\in\Delta$.

A stack content is a word from $\Gamma^*\bot$ where the leftmost 
symbol is assumed to be the top of the stack. A configuration is 
a pair $(q,\gamma)$ consisting of a state $q\in Q$ and a stack 
content $\gamma\in\Gamma^*\bot$. The stack height of a configuration 
$(q,\gamma)$ is defined by $\sh(q,\gamma)=|\gamma|-1$. 
Furthermore, we write $(q,\gamma)\mapstochar\relbar(q',\gamma')$ 
if there exists $(q,\gamma(0),q',\alpha)\in\Delta$ and 
$\gamma'=\alpha\gamma(1)\cdots \gamma(|\gamma|-1)$.

For a PDS $\mathcal P$, the induced pushdown graph is the infinite 
directed graph $G(\mathcal P)=(V, E)$ where 
$V=\{(q,\gamma)\mid q\in Q, \gamma\in\Gamma^*\bot\}$ is the set of 
configurations and $(v,v')\in E$ 
if $v\mapstochar\relbar v'$. Notice that every vertex of
the pushdown graph $G(\mathcal P)$ has at least one outgoing edge, 
since $\mathcal P$ is deadlock-free.
Consider a partition $Q_0 \cup Q_1$ of the set of states $Q$.
The induced pushdown game graph $G=(V,V_0,V_1,E,v_{\mathrm{in}})$ is a game graph where 
$(V, E)=G(\mathcal P)$, the partition $V_0\cup V_1$ of the set of configurations
$V$ is defined by $V_i=\{(q,\gamma)\in V\mid q\in Q_i\}$, for $i\in\{0,1\}$, 
and $v_{\mathrm{in}}=(q_{\mathrm{in}},\bot)$.
Given such a pushdown game graph $G$ and a coloring $\col\colon Q \rightarrow [n]$ of its states,
we obtain a parity game by extending $\col$ to configurations via $\col(q,\gamma)=\col(q)$,
for every state $q\in Q$ and every stack content $\gamma\in\Gamma^*\bot$.
We refer to such a game as a pushdown game.
%
%

\begin{example}\label{example_pushdown_game}
Consider the pushdown system $\mathcal P=(\{q_\mathrm{in},q_1,q_2\}, \{A\}, \Delta, q_\mathrm{in})$ where $\Delta$ is the following set 
\begin{equation*}\{
 (q_\mathrm{in},X, q_\mathrm{in}, AX),
 (q_\mathrm{in},X, q_1, AX),
 (q_1,A, q_1, \varepsilon),
 (q_1,\bot, q_2, \bot),
 (q_2,A, q_2, \varepsilon),
 (q_2,\bot, q_2, \bot)
\mid X\in \{A,\bot\}\}.\end{equation*} The partition $Q_0=\{q_1,q_2\}$ and $Q_1=\{q_\mathrm{in}\}$ yields 
the pushdown game graph $G$ depicted in Figure~\ref{pushdown_graph}, where the circles
indicate Player~$0$ configurations and squares are Player~$1$ configurations. With the coloring function $\col$ such that
$\col(q_\mathrm{in})=\col(q_2)=0$ and $\col(q_1)=1$ Player~$0$ wins the pushdown game $(G, \col)$, as every play visits only 
a finite number of configurations colored by $1$. 
\end{example}
\begin{figure}
\begin{center}
\begin{tikzpicture}[scale=1.1]

\foreach \x in {0,...,7}
  {
  \node[p1] at (1.3*\x,0) (a\x) {};
  \node[p0] at (1.3*\x,-.7) (b\x) {};
  }

  \node[p0] at (0,-1.4) (c0) {};

\node at (10.2,0) (aend) {$\cdots$};
\node at (10.2,-.7) (bend) {$\cdots$};

\node at (-1.5, 0) (alabel) {$q_\mathrm{in}$};
\node at (-1.5,-.7) (blabel) {$q_1$};
\node at (-1.5,-1.4) (clabel) {$q_2$};

\path[thick, draw]
(a0) edge (a1)
(a1) edge (a2)
(a2) edge (a3)
(a3) edge (a4)
(a4) edge (a5)
(a5) edge (a6)
(a6) edge (a7)
(a7) edge (aend)
(a0) edge (b1)
(a1) edge (b2)
(a2) edge (b3)
(a3) edge (b4)
(a4) edge (b5)
(a5) edge (b6)
(a6) edge (b7)
(a7) edge (bend)
(b1) edge (b0)
(b2) edge (b1)
(b3) edge (b2)
(b4) edge (b3)
(b5) edge (b4)
(b6) edge (b5)
(b7) edge (b6)
(bend) edge (b7)


(b0) edge[bend left=0] (c0)
(c0) edge[loop left] ()
(-.7,0) edge (a0);

\end{tikzpicture}

\end{center}
\caption{A pushdown game graph (only the part reachable from the initial vertex is shown)}
\label{pushdown_graph}
\end{figure}
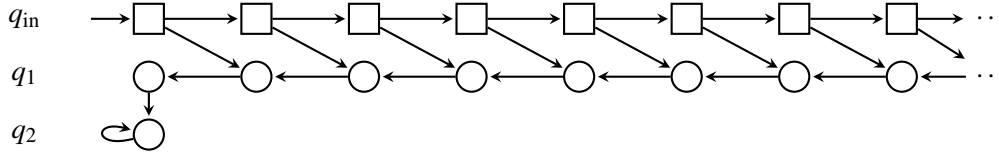

We extend the notion of PDS to pushdown transducers (PDT) by attaching input and
output alphabets. A PDT $\mathcal
T=(Q,\Gamma,\Delta,q_\mathrm{in},\Sigma_I,\Sigma_O, \lambda)$, where $Q$,
$\Gamma$ and $q_\mathrm{in}$ are as for PDS and $\Delta$ is modified such that
$\Delta\subseteq Q\times \Gammabot\times(\Sigma_I\cup\{\epsilon\})\times Q
\times \Gammabot^{\leq 2}$, additionally contains an input alphabet $\Sigma_I$,
an output alphabet $\Sigma_O$, and a partial output function $\lambda\colon Q
\rightarrow \Sigma_O$. A PDT is deterministic if it satisfies\
\[\left|\{(q',\alpha)\mid(q,A,a,q',\alpha)\in
\Delta\}\right|+\left|\{(q',\alpha)\mid(q,A,\epsilon,q',\alpha)\in
\Delta\}\right|\leq 1\]
for all $q\in Q$, all $a\in\Sigma_I$, and all $A\in\Gammabot$.
Analogously to PDS, we write
$(q,\gamma)\overset{a}{\mapstochar\relbar}(q',\gamma')$ if there exists a transition
$(q,\gamma(0),a,q',\alpha)\in\Delta$ such that $\gamma'=\alpha\gamma(1)\cdots
\gamma(|\gamma|-1)$. A run $\rho$ of a PDT on a word $w\in(\Sigma_I)^*$ is a
sequence of configurations $\rho=(q_0,\gamma_0)\cdots(q_m,\gamma_m)$ such that
$\rho(0)=(q_\mathrm{in},\bot)$, for all $0\leq i < m$ there exists
$a_i\in\Sigma_I\cup\{\epsilon\}$ with
$(q_i,\gamma_i)\overset{a_i}{\mapstochar\relbar}(q_{i+1},\gamma_{i+1})$ such that
$a_0\cdots a_{m-1}=w$, and $\{(q,\alpha)\mid(q_m,\gamma_m(0),\epsilon,q,\alpha)\in \Delta\}$
is empty (i.e., no execution of an $\epsilon$-transition is possible from the last
configuration of a run). A deterministic PDT $\mathcal T$ defines a partial
function $f_\mathcal{T}\colon (\Sigma_I)^* \rightarrow \Sigma_O$ such that
$f_\mathcal{T}(w)=\lambda(q)$, where $q$ is the state of the last configuration of
the (unique) run of $\mathcal T$ on $w$, if such a run exists.
    
To implement pushdown strategies in a pushdown game we will use PDT. To have a
finite input alphabet, we represent play prefixes here by sequences of
transitions and not by sequences of configurations. Notice that both
representations can easily be converted into each other. Furthermore, the output
will be the next transition to be chosen by Player~$i$ instead of the next
configuration. Hence, we use the set of transitions of the PDS defining the
pushdown game for both the input and the output alphabet of the PDT. So, the
transducer consumes a play prefix in the pushdown graph represented by a
sequence of transitions and outputs the transition which Player~$i$ should
choose next (in case the last configuration of the play prefix is a Player~$i$
configuration). Thus, we have to require the output transition to be executable
from the last configuration of the play prefix induced by the input sequence. 
\section{Finite-Time Pushdown Games}
\label{sec_scores}

In this section, we introduce a finite-duration variant of pushdown games.
To this end, we adapt the concept of scoring functions, which were originally
introduced by McNaughton \cite{DBLP:conf/hcat/McNaughton00} for Muller games 
(see also \cite{DBLP:journals/corr/abs-1006-1410}), to parity games. In the 
following, let $(G,\col)$ be a parity game with 
$G=(V,V_0,V_1,E,v_{\mathrm{in}})$ and $\col\colon V \rightarrow [n]$.
\begin{definition}[Scoring functions]
For every $c\in[n]$, define the function $\Sc_c\colon V^* \rightarrow \nats$ 
by $\Sc_c(\varepsilon)=0$ and for $w\in V^*$ and $v\in V$ by
\[
\Sc_c(wv)=
\begin{cases} 
\Sc_c(w) & \text{if }\col(v)>c,\\  
\Sc_c(w)+1 & \text{if } \col(v)=c, \\
0 & \text{if } \col(v)<c. 
\end{cases}
\]
Furthermore, for every $c\in [n]$,  
$\MaxSc_c\colon V^*\cup V^\omega \rightarrow \nats\cup \{\infty\}$ 
is defined by 
$\MaxSc_c(\rho)=\underset{w\sqsubseteq \rho}{\Max} \ \Sc_c(w)$.
\end{definition}

A positional winning strategy for Player~$i$ in a parity game
does not visit a vertex~$v$ with $\Parity(\col(v))=1-i$ twice without 
visiting some vertex of strictly smaller color in between.
Hence, applying the pigeonhole principle shows that positional winning 
strategies in finite parity games bound the scores of the losing player.
For $c\in[n]$, let $|V|_c$ denote the number of vertices of color $c$,
i.e., $|V|_c=|\{v\in V \mid \col(v)=c\}|$.

\begin{remark}
\label{remark}
Let $\strat$ be a positional winning strategy for Player~$i$ in a parity game
with a finite vertex set~$V$. Then, for every $\rho$ that is consistent with
$\strat$, $\MaxSc_c(\rho)\leq |V|_c$ for all $c\in[n]$ such that
$\Parity(c)=1-i$.
\end{remark}

Thus, winning a finite parity game, i.e., a parity game with a finite game graph, can also be characterized by being able to
achieve a certain threshold score. As soon as this threshold score is reached
the play can be stopped, since the winner is certain. This is the idea behind
finite-time versions of infinite games. Formally, a finite-time parity game $(G,
\col, k)$ consists of a game graph $G$, a min-parity condition $\col$ and a
threshold $k\in \nats \setminus \{ 0 \}$. A play in $(G, \col, k)$ is a finite path
$w=w(0)\cdots w(r) \in V^*$ with $w(0)=v_\mathrm{in}$ such that $\MaxSc_c(w)=k$
for some $c\in[n]$, and $\MaxSc_{c}(w(0)\cdots w(r-1))<k$ for all $c\in [n]$.
The play $w$ is winning for Player~$i$ if $\Parity(c)=i$. The notions of
(winning) strategies are defined as usual.


By induction over ther number $n$ of colors one can show that every threshold $k$ is eventually reached by some score
function if the path is sufficiently long. Thus, there are no
draws due to infinite plays.

\begin{lemma}
\label{every_k_is_reached} 
For every $w\in V^*$ with $|w|\geq k^{n}$, there is some $c\in[n]$ such that
$\MaxSc_c(w)\geq k$. 
\end{lemma}

Hence, a play in a finite-time parity game is stopped after at most
exponentially many rounds. Moreover, using the construction of \cite{DBLP:journals/corr/abs-1006-1410} for Muller games (which also holds for parity games) 
it can also be shown that the bound in
Lemma~\ref{every_k_is_reached} is tight, i.e., for every $k$ there is a $w\in
V^*$ with $|w|=k^{n}-1$ such that $\MaxSc_c(w)< k$ for all $c\in[n]$.

Furthermore, it is never the case that two different score functions are
increased  in the same round: by definition of the score functions, only the
value of $\Sc_{\col(w(i))}$ is increased in round $i$ of a play $w$. Hence, as
soon as some score function is increased to the threshold a unique winner
can be declared.

\begin{lemma}
\label{no_2_increase_in_the_same_round}
Let $w\in V^*$, $v\in V$ and $c,c'\in [n]$. If $\Sc_c(wv)=\Sc_c(w)+1$ 
and $\Sc_{c'}(wv)=\Sc_{c'}(w)+1$, then $c=c'$. 
\end{lemma}

In \cite{DBLP:journals/corr/abs-1006-1410}, the equivalence 
between Muller games and finite-time 
Muller games (using the original scoring functions for Muller games) on
finite game graphs is shown for the constant threshold $k=3$.
A simple consequence of Remark~\ref{remark} is an analogous 
result for parity games on finite game graphs.

\begin{theorem}
\label{thm_finite-arena_player_wins_iff_he_wins_ft}
Let $G$ be a finite game graph with vertex set $V$ and $\col\colon V \rightarrow
[n]$. For every threshold~$k>\underset{c\in [n]}{\Max} \ |V|_c$,
Player~$i$ wins $(G,\col)$ if and only if Player~$i$ wins $(G,\col,k)$.
\end{theorem}

It is easy to see that this result does not hold for infinite game graphs.
Consider the pushdown game from Example~\ref{example_pushdown_game} and
recall that Player~$0$ wins it. However, for every threshold $k>0$, Player~$1$
has a winning strategy in the corresponding finite-time pushdown game by moving
the token to configuration $(q_1,A^{k-1}\bot)$, which completely specifies a
strategy for Player~$1$. Following this strategy, Player~$1$ wins since color
$1$ is the first to reach score $k$ which happens when the token arrives at the
configuration~$(q_1,\bot)$.  

To obtain an analogous result for pushdown games, we have to adapt the
scoring functions. Now, let $(G,\col)$ be a pushdown game. Fix a path through
the pushdown graph. A configuration is said to be a stair configuration, if no
subsequent configuration of smaller stack height exists in this path.

\begin{definition}[Stairs \cite{DBLP:conf/fsttcs/LodingMS04}]
Define the functions $\StairPositions\colon V^+ \cup V^\omega \rightarrow 2^\nats$ and 
$\Stairs\colon V^+ \cup V^\omega \rightarrow V^+ \cup V^\omega$ as follows: for $w\in V^+ \cup V^\omega$, let
\[\StairPositions(w)=\{n\in \nats \mid \forall m\geq n: \sh(w(m)) \geq \sh(w(n)) \} \]
and $\Stairs(w)=w(n_0)w(n_1)\cdots$, where $n_0 < n_1 < \cdots$ is the ascending enumeration of
$\StairPositions(w)$.
\end{definition}

Now, using the notion of stairs, we define  stair-score functions for pushdown
games. To simplify our notation, let $\Reset(v)=\varepsilon$ and
$\lastBump(v)=v$ for $v\in V$ and for $w=w(0)\cdots w(r)$ with $r\geq 1$, let
$\Reset(w)=w(0)\cdots w(l)$ and $\lastBump(w)=w(l+1)\cdots w(r)$, where $l$ is
the greatest position such that $\sh(w(l))\leq \sh(w(r))$ and $l\neq r$, i.e.,
$l$ is the second largest\footnote{Notice that the last position of a finite
path is always a stair position.} stair position of $w$. Figure~\ref{fig:stairs}
illustrates the above definitions, where an example path $w$ and the
corresponding stack heights are depicted. The stair positions are indicated by
the marked stack heights. Furthermore, the figure also illustrates our new
definition of stair-scores which we define next. 
\begin{figure}
\begin{center}
\begin{tikzpicture}[scale=.95]
 
\draw[fill=black] (0,0)    circle (.09cm);      \draw[thick, draw=black, fill=none] (0,0) circle (.16cm); 
\draw[fill=black] (1,0.5)  circle (.09cm);      \draw[thick, draw=black, fill=none] (1,0.5) circle (.16cm); 
\draw[fill=black] (2,1)    circle (.09cm);     \draw[thick, draw=black, fill=none] (2,1) circle (.16cm); 
\draw[fill=black] (3,1.5)  circle (.09cm); 
\draw[fill=black] (4,2)    circle (.09cm); 
\draw[fill=black] (5,1.5)  circle (.09cm); 
\draw[fill=black] (6,2)    circle (.09cm);        
\draw[fill=black] (7,1.5)  circle (.09cm);      
\draw[fill=black] (8,1)    circle (.09cm);     \draw[thick, draw=black, fill=none] (8,1) circle (.16cm); 
\draw[fill=black] (9,1.5)  circle (.09cm);     \draw[thick, draw=black, fill=none] (9,1.5) circle (.16cm); 
\draw[fill=black] (10,1.5) circle (.09cm);     \draw[thick, draw=black, fill=none] (10,1.5) circle (.16cm); 
\draw[fill=black] (11,2)   circle (.09cm);       
\draw[fill=black] (12,2.5) circle (.09cm); 
\draw[fill=black] (13,2)   circle (.09cm); 
\draw[fill=black] (14,1.5) circle (.09cm);     \draw[thick, draw=black, fill=none] (14,1.5) circle (.16cm); 
 
\path[draw, thick] (0,0) -- (1,0.5) -- (2,1) -- (3,1.5) -- (4,2) -- (5,1.5) -- (6,2) -- (7,1.5) -- (8,1) -- (9,1.5) -- (10,1.5) -- (11,2) -- (12,2.5) -- (13,2) -- (14,1.5); 
 
\foreach \x in {0, 0.5,...,2.5} 
\path[draw, dashed] (-.3, \x) -- (14,\x); 
 
\path[draw, -stealth, thick] (0,-.4) -- (14.3,-.4); 

\path[draw, -stealth, thick] (-.3,0) -- (-.3,3.0); 
 
\draw[decorate, thick, decoration={brace,amplitude=5}] (-.2,2.7) -- node[auto, yshift =3]{$\Reset(w)$} (10.3,2.7); 
 
\draw[decorate, thick, decoration={brace,amplitude=5}] (10.7,2.7) -- node[auto, yshift =3]{$\lastBump(w)$} (14.3,2.7); 
 
\draw[draw=none] (-0.5,0) -- node[sloped]{stack height} (-0.5,3); 
 
\node at (14.65,-.4) {$w$}; 
 
\node at (-1,-.8) {col}; 
\node at (0,-.8) {$0$}; 	
\node at (1,-.8) {$2$}; 
\node at (2,-.8) {$1$}; 
\node at (3,-.8) {$0$}; 
\node at (4,-.8) {$2$}; 
\node at (5,-.8) {$1$}; 
\node at (6,-.8) {$0$}; 
\node at (7,-.8) {$0$}; 
\node at (8,-.8) {$0$}; 
\node at (9,-.8) {$1$}; 
\node at (10,-.8) {$1$}; 
\node at (11,-.8) {$1$}; 
\node at (12,-.8) {$1$}; 
\node at (13,-.8) {$2$}; 
\node at (14,-.8) {$1$}; 
 

\node at (-1,-1.3) {$\StairSc_0$}; 
\node at (0,-1.3) {$1$}; 
\node at (1,-1.3) {$1$}; 
\node at (2,-1.3) {$1$}; 
\node at (3,-1.3) {$2$}; 
\node at (4,-1.3) {$2$}; 
\node at (5,-1.3) {$2$}; 
\node at (6,-1.3) {$3$}; 
\node at (7,-1.3) {$3$}; 
\node at (8,-1.3) {$2$}; 
\node at (9,-1.3) {$2$}; 
\node at (10,-1.3) {$2$}; 
\node at (11,-1.3) {$2$}; 
\node at (12,-1.3) {$2$}; 
\node at (13,-1.3) {$2$}; 
\node at (14,-1.3) {$2$};

 
\node at (-1,-1.8) {$\StairSc_1$}; 
 
\node at (0,-1.8) {$0$}; 
\node at (1,-1.8) {$0$}; 
\node at (2,-1.8) {$1$}; 
\node at (3,-1.8) {$0$}; 
\node at (4,-1.8) {$0$}; 
\node at (5,-1.8) {$1$}; 
\node at (6,-1.8) {$0$}; 
\node at (7,-1.8) {$0$}; 
\node at (8,-1.8) {$0$}; 
\node at (9,-1.8) {$1$}; 
\node at (10,-1.8) {$2$}; 
\node at (11,-1.8) {$3$}; 
\node at (12,-1.8) {$4$}; 
\node at (13,-1.8) {$4$}; 
\node at (14,-1.8) {$3$}; 
 
 
\node at (-1,-2.3) {$\StairSc_2$}; 
 
\node at (0,-2.3) {$0$}; 
\node at (1,-2.3) {$1$}; 
\node at (2,-2.3) {$0$}; 
\node at (3,-2.3) {$0$}; 
\node at (4,-2.3) {$1$}; 
\node at (5,-2.3) {$0$}; 
\node at (6,-2.3) {$0$}; 
\node at (7,-2.3) {$0$}; 
\node at (8,-2.3) {$0$}; 
\node at (9,-2.3) {$0$}; 
\node at (10,-2.3) {$0$}; 
\node at (11,-2.3) {$0$}; 
\node at (12,-2.3) {$0$}; 
\node at (13,-2.3) {$0$}; 
\node at (14,-2.3) {$0$}; 
\end{tikzpicture}
\end{center}
  \caption{A finite path w, its stair positions and its stair-scores.}
\label{fig:stairs} 
\end{figure}
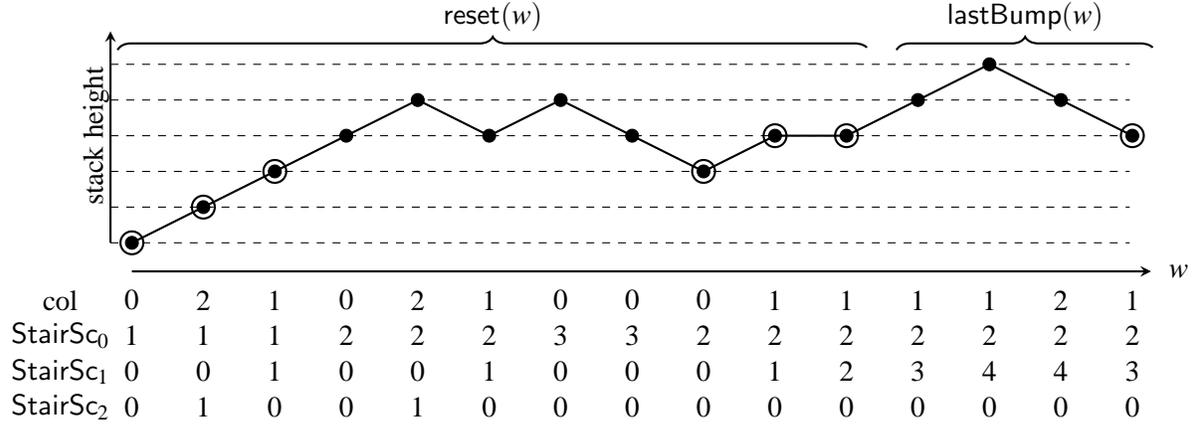

\begin{definition}[Stair-scoring function]
For every color $c\in[n]$, define the function $\StairSc_c\colon V^* \rightarrow \nats$ by
$\StairSc_c(\varepsilon)=0$ and for $w\in V^+$ by 
\[
\StairSc_c(w)=
\begin{cases}
\StairSc_c(\Reset(w)) & \text{if } \MinCol(\lastBump(w))>c,\\
\StairSc_c(\Reset(w))+1 & \text{if } \MinCol(\lastBump(w))=c,\\
0 & \text{if } \MinCol(\lastBump(w))<c.
\end{cases}
\]
\noindent Furthermore, for every color $c\in[n]$, the function
$\MaxStairSc_c\colon V^* \cup V^\omega \rightarrow \nats \cup \{\infty\}$ 
is defined by $\MaxStairSc_c(\rho)=\underset{w\sqsubseteq \rho}{\Max} \ \StairSc_c(w)$.
\end{definition}

Now, using these notions we define finite-time pushdown games. Such a game~$(G,
\col, k)$ consists of a pushdown game graph $G$, a min-parity condition $\col$
and a threshold $k\in \nats\setminus\{0\}$. A play in $(G, \col, k)$ is a finite path
$w=w(0)\cdots w(r) \in V^*$ with $w(0)=v_\mathrm{in}$ such that
$\MaxStairSc_c(w)=k$ for some $c\in[n]$, and $\MaxStairSc_{c}(w(0)\cdots
w(r-1))<k$ for all $c\in [n]$. The play $w$ is winning for Player~$i$ if
$\Parity(c)=i$. Again, the notions of (winning) strategies are defined as usual.

As above, every threshold~$k$ is eventually reached by some stair-score
function if the play is sufficiently long: a  simple
induction shows that every $w\in V^+$ with $0\in\StairPositions(w)$ and with
$|w|\geq 2^m$ has a prefix $w'\sqsubseteq w$ such that $|\StairPositions(w')| >
m$. Furthermore, for every play prefix $w'\sqsubseteq w$ a sequence $u'\in Q^*$
of states with $|u'|=|\Stairs(w')|$ can be constructed such that for every color
$c\in[n]$, $\StairSc_c(w')=\Sc_c(u')$. Combining these two properties and
Lemma~\ref{every_k_is_reached} yields the desired upper bound on the length of a play.

\begin{lemma}
\label{every_k_is_reached_pushdownversion} 
For every $w\in V^*$ with
$|w|\geq 2^{k^{n}}$ there is some 
$c\in[n]$ such that $\MaxStairSc_c(w)\geq k$. 
\end{lemma}

Thus, a play in a finite-time pushdown game stops after a doubly-exponential
number of rounds. Again, the bound in
Lemma~\ref{every_k_is_reached_pushdownversion} is tight. Moreover,
Lemma~\ref{no_2_increase_in_the_same_round} can directly be translated to the
new definition of stair-scoring functions which ensures a unique winner of a
play.

\begin{lemma}
Let $w\in V^*$, $v\in V$ and $c,c'\in [n]$. If $\StairSc_c(wv)=\StairSc_c(w)+1$ 
and $\StairSc_{c'}(wv)=\StairSc_{c'}(w)+1$, then $c=c'$. 
\end{lemma}

In Section~\ref{sec_main}, we prove the equivalence between pushdown games and
finite-time pushdown games. To this end, we adapt Walukiewicz's reduction from
pushdown parity games to parity games on finite game graphs, which we recall in
the following section.

\section{Walukiewicz's Reduction}
\label{sec_walu}

Walukiewicz showed that pushdown games can be solved 
in exponential time \cite{DBLP:conf/cav/Walukiewicz96}. 
In this section, we recall his technique which comprises 
a reduction to parity games on finite game 
graphs. We present a slight modification of 
the original construction which is needed to prove our 
result in the next section.

Let $\mathcal G=(G,\col)$ 
be a pushdown game with game graph 
$G=(V,V_0,V_1,E,v_\mathrm{in})$ induced by $\mathcal P=(Q,\Gamma,\Delta,q_\mathrm{in})$ with 
partition $Q_0\cup Q_1$ of $Q$ and min-parity condition 
$\col\colon Q\rightarrow [n]$. To simulate $\mathcal G$ by a 
game on a finite game graph the information stored on the stack 
is encoded by some finite memory structure. The essential 
component of this structure is the set 
$\Prediction=(\Powerset(Q))^{n}$, which we call the set of 
predictions. A prediction 
$P=(P_0, \ldots, P_{n-1})\in\Prediction$ contains for every 
$c\in[n]$ a subset $P_c\subseteq Q$ of states.

The core idea of the game simulating the 
pushdown game is the following:
The players are assigned different tasks, one of them makes 
predictions and the other one verifies them. Whenever a 
push-transition is to be simulated the predicting player 
has to make a prediction $P\in\Prediction$ about the future round $t$ when the same stack height as 
before performing the push-transition is reached again for 
the first time (if it is reached at all). 
With this prediction, the predicting player claims that if the current push-transition is performed, 
then in round $t$ some state $q\in P_c$ will be reached if $c\in[n]$ 
is the minimal color seen in between. Once a prediction $P$ is 
proposed, the verifying player has two ways of reacting, 
either believing that $P$ is correct or not. In the first case, 
he is not interested in verifying $P$, so the push-transition is 
not performed and the verifying player chooses a color $c\in[n]$ 
and a state $q\in P_c$, for some $P_c\neq\emptyset$, and skips a part 
of the simulated play by jumping to an appropriate position in 
the play. In the other case, he wants to verify the correctness 
of $P$, so the push-transition is performed and when the top of the 
stack is eventually popped it will turn out whether $P$ is correct 
or not. The predicting player wins if $P$ turns out to be correct 
and otherwise the verifying player wins. So after a pop-transition 
the winner is certain. For the other case, 
where no pop-transition is performed at all, the parity condition 
determines the winner. 

In the following, let Player~$i$ take the 
role of the predicting player and Player~$1-i$ the role of the verifying one.
The game $\mathcal G_i'=(G',\col')$ which depends on $i\in\{0,1\}$, with 
$G'=(V',V_0',V_1',E', v_\mathrm{in}')$ is defined as follows: 
For all states $q\in Q$, stack symbols $A, B \in \Gammabot$, 
colors $c,d\in[n]$ and predictions $P,R\in \Prediction$, the set 
$V'$ contains the vertices $\WaluCheck[q,A,P,c,d]$ 
which correspond to the configurations of $\mathcal G$, auxiliary 
vertices $\WaluPush[P,c,q,AB]$, $\WaluClaim[P,c,q,AB,R]$ and 
$\WaluJump[q,A,P,c,d]$ which serve as intermediates to signalize 
the intention to perform a push-transition, to make a new prediction 
and to skip a part of a simulated play, and finally the sink vertices 
$\WaluWin_i[q]$ and $\WaluWin_{1-i}[q]$. 

The set $E'$ consists of the following edges (for the sake of readability, we
denote an edge $(v_1,v_2)\in E'$ here by $v_1 \rightarrow v_2$).
For every skip-transition $\delta=(q,A,p,B)\in\Delta$ there are edges 
\[\WaluCheck[q,A,P,c,d]\rightarrow\WaluCheck[p,B,P,\Min\{c,\col(p)\},\col(p)]\, ,\]
for $P\in \Prediction$ and $c,d\in[n]$. Thus, the first 
two components of the $\WaluCheck$-vertices are updated according to $\delta$, 
the prediction $P$ remains untouched, the last but one component is used 
to keep track of the minimal color for being able to check the prediction 
for correctness and the last component determines the color of the current 
$\WaluCheck$-vertex. 
For every push-transition $\delta=(q,A,p,BC)\in\Delta$ there are edges 
\[\WaluCheck[q,A,P,c,d]\rightarrow\WaluPush[P,c,p,BC]\, ,\] 
for all $P\in \Prediction$ and $c,d\in[n]$. 
Here, a player states that a push-transition is to be performed such 
that the current state $q$ has to be changed to $p$
and the top of the stack $A$ has to be replaced by $BC$. The information 
containing the current prediction $P$ and the minimal color $c$ is 
carried over, as this is needed in the case where the verifying player 
decides to skip.
Moreover, to make a new prediction~$R$, all edges
\[\WaluPush[P,c,p,BC]\rightarrow\WaluClaim[P,c,p,BC,R]\] for every $R \in \Prediction$ are needed. 
In case a new prediction is to be verified, a push-transition is 
finally performed 
using edges of the form 
\[\WaluClaim[P,c,p,BC,R]\rightarrow\WaluCheck[p,B,R,\col(p),\col(p)]\]
where the prediction $P$, the 
color $c$ and the lower stack symbol $C$ are discarded, since they are no longer needed. For the other 
case, where the verifying player intends to skip a part of a play, all edges 
\[\WaluClaim[P,c,p,BC,R]\rightarrow\WaluJump[q,C,P,c,e] \]
with  $q\in R_e$ are contained in $E'$. Here, the verifying player chooses a 
color $e\in[n]$ for the minimal color of the skipped part and a state $q$ from 
the corresponding component $R_e$ of the prediction $R$. Now, the lower stack symbol~$C$, 
the prediction $P$ and the color $c$ additionally have to be carried over, 
whereas $B$ and $R$ are discarded. Then, all edges 
\[\WaluJump[q,C,P,c,e]\rightarrow\WaluCheck[q,C,P,\Min\{c,e,\col(q)\},\Min\{e,\col(q)\}]\] 
are contained in $E'$ where the last component of the $\WaluCheck$-vertex 
is set to be the minimum of the color of the current state $q$ and the 
minimal color of the part just skipped. For the last but one component, we also 
have to account for the color $c$, which is necessary for eventually checking $P$ 
for correctness. 
Finally, we have for every pop-transition 
$(q,A,p,\varepsilon)\in\Delta$, the edges 
\begin{align*}
&\WaluCheck[q,A,P,c,d]\rightarrow\WaluWin_i[p]\enspace\text{ if }p\in P_c \, , \text{ and} \\
&\WaluCheck[q,A,P,c,d]\rightarrow\WaluWin_{1-i}[p]\enspace\text{ if }p\notin P_c\, ,
\end{align*} 
for $P\in\Prediction$ and $c,d\in[n]$, 
which lead to the sink vertex of the predicting player $\WaluWin_i[p]$ if the 
prediction~$P$ turns out to be correct or to the sink vertex of the verifying 
player $\WaluWin_{1-i}[p]$ otherwise. Moreover, we have
$(\WaluWin_{j}[q],\WaluWin_{j}[q])\in E'$, for $j\in\{0,1\}$ and $q\in Q$.

The initial vertex $v_\mathrm{in}'$ has to correspond 
to the initial configuration $v_\mathrm{in}=(q_\mathrm{in},\bot)$, so it is defined to be 
$\WaluCheck[q_\mathrm{in}, \bot, P^\mathrm{in}, \col(q_\mathrm{in}),\col(q_\mathrm{in})]$ 
where $P^\mathrm{in}_c=\emptyset$ for every $c\in[n]$, as the 
$\bot$-symbol cannot be deleted from the stack. The set of vertices 
$V_i'$ of the predicting Player~$i$ is defined to consist 
of all $\WaluPush$-vertices, as there Player~$i$ has to make 
a new prediction, and of those $\WaluCheck[p,A,P,c,d]$ vertices where $p\in Q_i$. 
Accordingly, all other vertices belong to Player~$1-i$.
Finally, the coloring function $\col'\colon V' \rightarrow [n+1]$ is defined by 
$\col'(\WaluCheck[p,A,P,c,d])=d$ and 
$\col'(\WaluWin_j[q])=j$, for 
$j\in\{0,1\}$. All other vertices are colored by the maximal color $n$ (which does not appear in $\mathcal G$), since 
they are auxiliary vertices and should have no influence on the minimal color seen infinitely often. 
This is guaranteed by the structure of $G'$, as there are no loops consisting only of auxiliary vertices. 
Notice that in the original construction, $\WaluJump$-vertices
are colored by the minimal color of the skipped part of the play 
which is chosen by the verifying player. This is avoided here by shifting 
the color of a $\WaluJump$-vertex to the successive $\WaluCheck$-vertex. 
For this purpose, the last component of the $\WaluCheck$-vertices is introduced.

\begin{theorem}[\cite{DBLP:conf/cav/Walukiewicz96}]
\label{thm_walu}
Let $\mathcal G$ be a pushdown game. Player~$i$ wins 
$\mathcal G$ if and only if Player~$i$ wins $\mathcal G_i'$.
\end{theorem}

Now, let us describe how a winning strategy $\strat$ for Player~$i$ in $\mathcal G$  
can be constructed from a positional winning strategy $\strat_i'$ for Player~$i$ 
in $\mathcal G_i'$.
The idea is to simulate $\strat_i'$ in $\mathcal G$. This works out fine as long as only skip- and push-transitions
are involved. As soon as the first pop-transition is used, $\strat_i'$ leads to a sink $\WaluWin_i$-vertex 
at which the future moves of $\strat_i'$ are no longer useful for playing in the original game $\mathcal G$.
To overcome this, the strategy $\strat$ uses a stack to store $\WaluClaim$-vertices visited during the simulated play. 
This allows us to reset the simulated play and to continue from the appropriate successor $\WaluJump$-vertex of the 
$\WaluClaim$-vertex stored on the stack.

Formally, let $G'|_{\strat_i'}=(V'|_{\strat_i'},V_0'|_{\strat_i'},V_1'|_{\strat_i'},E'|_{\strat_i'},v_\mathrm{in}')$
be the game graph of $\mathcal G_i'$ restricted to the vertices and edges visited by $\strat_i'$.
This implies that every vertex from $V_i'|_{\strat_i'}$ 
has a unique successor in $G'|_{\strat_i'}$ and that $\WaluWin_{1-i}$-vertices are not contained in $V_i'|_{\strat_i'}$.
The pushdown transducer $\mathcal T_\strat$ implementing $\strat$ is obtained from $\strat_i'$ 
by employing $G'|_{\strat_i'}$ for its finite control and 
the $\WaluClaim$-vertices 
as its stack symbols.

The PDT implementing $\strat$ is defined by $\mathcal T_\strat=(Q^\strat,\Gamma^\strat,\Delta^{\!\!\strat},q_\mathrm{in}^\strat,\Sigma_I^\strat,\Sigma_O^\strat, \lambda^{\!\!\!\strat})$,
where $Q^\strat=V'|_{\strat_i'}$, $\Gamma^\strat=\{v\in V'|_{\strat_i'} \ \mid v\text{ is a $\WaluClaim$-vertex}\}$, 
$q_\mathrm{in}^\strat=v_\mathrm{in}'$, $\Sigma_I^\strat=\Sigma_O^\strat=\Delta$.
To define $\Delta^{\!\!\strat}$, we first define the labeling $\ell\colon E'|_{\strat_i'} \rightarrow \Delta\cup\{\epsilon\}$
which assigns to every edge in $E'|_{\strat_i'}$ its corresponding transition $\delta\in\Delta$ by 
\[\ell(v,v')=
\begin{cases}
 (q,A,p,B) & \text{if } (v,v')=(\WaluCheck[q,A,P,c,d],\WaluCheck[p,B,P,c',d'])\, ,\\
 (q,A,p,BC) & \text{if } (v,v')=(\WaluCheck[q,A,P,c,d],\WaluPush[P,c,p,BC])\, ,\\
 (q,A,p,\epsilon) & \text{if } (v,v')=(\WaluCheck[q,A,P,c,d],\WaluWin_i[p])\, ,\\
 \epsilon & \text{otherwise.}
\end{cases}
\]
Now, the transition relation $\Delta^{\!\!\strat}$ is defined as follows: for every $(v,v')\in E'|_{\strat_i'}$, if $v$ is not a $\WaluClaim$-vertex
\emph{and}
$v'$ is not a $\WaluWin_i$-vertex, then $(v,Z,\ell(v,v'),v',Z)\in\Delta^{\!\!\strat}$, for every $Z\in\Gammabot^\strat$.
For the other cases, if $v$ is a $\WaluClaim$-vertex and $v'$ is a $\WaluCheck$-vertex, then $(v,Z,\ell(v,v'),v',vZ)\in\Delta^{\!\!\strat}$ 
for $Z\in\Gammabot^\strat$, i.e., the $\WaluClaim$-vertex $v$ is pushed onto the stack. And finally, if $(v,v')=(\WaluCheck[q,A,P,c,d],\WaluWin_i[p])$,
then $(v,Z, \ell(v,v'), \WaluJump[p,C,R,e,c], \epsilon)\in\Delta^{\!\!\strat}$ for every $Z\in\Gamma^\strat$ of the form $Z=\WaluClaim[R,e,q',BC,R']$, i.e.,
the topmost symbol $\WaluClaim[R,e,q',BC,R']$ is popped from the stack and the pushdown transducer proceeds to the state $\WaluJump[p,C,R,e,c]$
which would be reached in $G_i'|_{\strat_i'}$ if Player~$1-i$ 
would have chosen color $c$ and state $p\in R_c$ to determine the successor of $\WaluClaim[R,e,q',BC,R']$.
To complete the definition of $\mathcal T_\strat$, we define the output function $\lambda^{\!\!\!\strat}$ by 
$\lambda^{\!\!\!\strat}(v)= \ell(v,v')$ if $v\in V_i'|_{\strat_i'}$ is a $\WaluCheck$-vertex and  
$(v, v')\in E'|_{\strat_i'}$, i.e., the labeling of the edge chosen 
by $\strat_i'$ determines the output of $\mathcal T_\strat$.
Lemma~\ref{walu_scores_are_equivalent_to_stair_scores} shows this construction to be correct. 
\section{Main Theorem}
\label{sec_main}

In this section, we prove the equivalence between a 
pushdown game and the corresponding finite-time 
pushdown game for a certain threshold which is 
exponential in the size of the PDS defining the 
pushdown game. 
For a pushdown game $\mathcal G=(G,\col)$ induced by
$\mathcal P=(Q,\Gamma,\Delta, q_\mathrm{in})$ and 
$\col\colon Q \rightarrow [n]$, define 
$k_\mathcal{G}=|Q|\cdot|\Gamma|\cdot2^{|Q|\cdot n}\cdot n$, 
which is an upper bound on the number of $\WaluCheck$-vertices 
in $\mathcal G_i'$ of the same color.

\begin{theorem}
\label{thm_main}
Let $\mathcal G=(G, \col)$ be a pushdown game and let
$\mathcal G_k=(G, \col, k)$ be the corresponding 
finite-time pushdown game with threshold $k$. For every
$k>k_\mathcal{G}$, Player~$i$ wins $\mathcal G$ 
if and only if Player~$i$ wins~$\mathcal G_k$.
\end{theorem}

To prove this theorem, we need the following lemma which establishes a relation
between the values of the scoring functions of plays in $\mathcal G_i'$
and the values of the stair-scoring functions of corresponding plays in $\mathcal G$. 
Let $\strat_i'$ be a positional winning strategy for Player $i$ in $\mathcal G_i'$
and $\mathcal T_\strat$ the PDT implementing the corresponding pushdown winning strategy $\strat$ for Player $i$ in 
$\mathcal G$ as defined in the previous section. For a play prefix $w(0)\cdots w(r)\in V^+$,
define $\lastStrictBump(w)=w$ if $\sh(w(r))=0$, and otherwise
$\lastStrictBump(w)=w(l+1)\cdots w(r)$ where $l$ 
is the greatest position such that $\sh(w(l))< \sh(w(r))$.

\begin{lemma}
\label{walu_scores_are_equivalent_to_stair_scores}
For every play prefix~$w$ in $\mathcal G$ that is consistent with $\sigma$, 
there is a play prefix~$w'$ in $\mathcal G_i'$ that is consistent with 
$\sigma_i'$ such that $\StairSc_c(w) = \Sc_c(w')$ for every $c \in [n]$.
\end{lemma}

\begin{proof}
By induction over $|w|$. To prove our claim, we strengthen the 
induction hypothesis as follows: for every play prefix~$w$ in $\mathcal G$ 
that is consistent with $\sigma$, there is a play prefix~$w'$ in 
$\mathcal G_i'|_{\strat_i'}$ (which is consistent with $\strat_i'$ by construction) such that the following 
requirements are satisfied: let $\last(w) = (q, A\gamma)$.

\begin{enumerate}[(i)]

\item\label{req_score} $\StairSc_c(w) = \Sc_c(w')$ for every $c \in [n]$.

\item\label{req_lastvert} $\last(w')=\WaluCheck[q, A, P, c, d]$ for some $P\in\Prediction$, $d\in[n]$ and $c=\MinCol(\lastStrictBump(w))$.

\item\label{req_stack} Let $(v, \gamma_\strat)$ be the last 
configuration of the run of $\mathcal T_\strat$ on the sequence of 
transitions induced by~$w$. Furthermore, if $\gamma_\strat\neq\bot$, let 
$\gamma_\strat(j) = \WaluClaim[P_j, c_j, p_j, B_jC_j, R_j]$ for every $0 \le j \le |\gamma_\strat|-2$.
We require $v = \last(w')$, $C_0 \cdots C_k=\gamma$ where $k=|\gamma_\strat|-2$, and $R_0=P$.


\end{enumerate}

For the induction start, we have $w = v_{\mathrm{in}} =(q_{\mathrm{in}}, \bot )$. 
Let $w' = v_\mathrm{in}' = \WaluCheck[q_\mathrm{in}, \bot, P^\mathrm{in}, \col(q_\mathrm{in}),\col(q_\mathrm{in})]$. 
Since $\col (v_\mathrm{in}) = \col'(v_\mathrm{in}') = \col(q_\mathrm{in})$, we have $\StairSc_c(w) = \Sc_c(w')$ for every $c \in [n]$. Moreover, requirements 
(\ref{req_lastvert}) and (\ref{req_stack}) are satisfied as well.

Now, let $w = w(0)\cdots w(r)$ with $r>0$ and $w(r-1) = (q, A\gamma)$. Moreover, let $\Reset(w)=w(0)\cdots w(s)$ and $w(s) = (q_s, A_s\gamma_s)$.  
The induction hypothesis yields play prefixes $u'$ and $u_s'$ in $\mathcal G_i'|_{\strat_i'}$ 
such that we have
$\StairSc_c(w(0)\cdots w(r-1)) = \Sc_c(u')$ and $\StairSc_c(w(0)\cdots w(s)) = \Sc_c(u_s')$, for every $c \in [n]$. 
Also, for some $P,P_s\in\Prediction$ and $d,d_s\in[n]$,
$\last(u')=\WaluCheck[q, A, P, c, d]$ and 
$\last(u_s')=\WaluCheck[q_s, A_s, P_s, c_s, d_s]$ with 
$c=\MinCol(\lastStrictBump(w(0)\cdots w(r-1)))$ and $c_s=\MinCol(\lastStrictBump(w(0)\cdots w(s)))$. 
We distinguish three cases, whether the transition from $w(r-1)$ 
to $w(r)$ is a skip-, push-, or pop-transition.

In case of a skip-transition~$\delta=(q, A, p, B)$, we have 
$w(r) = (p, B\gamma)$.
By construction, there is also an edge from 
$\last(u')=\WaluCheck[q, A, P, c, d]$ to the vertex 
\[v = \WaluCheck[p, B, P, \Min\{c,\col(p)\}, \col(p)]\] in $\mathcal{G}_i'|_{\strat_i'}$
labeled by $\ell(\last(u'),v)=\delta$. 
Thus, let $w' = u'v$. This choice satisfies requirement (\ref{req_lastvert}), as for a skip-transition from $w(r-1)$ to $w(r)$ it holds
\begin{align*}
\MinCol(\lastStrictBump(w)) & = \Min\{\MinCol(\lastStrictBump(w(0)\cdots w(r-1))),\col(w(r))\} \\
 & = \Min\{c,\col(p)\}\enspace.
\end{align*}
Furthermore, requirement (\ref{req_stack}) is satisfied, since when processing $\delta$, 
$\mathcal T_\strat$ changes its state $\last(u')$ to $v$ while the stack is left unchanged.
To prove the equality of the scores, let $e = \col(w(r))$, 
which is also the color of $v$ in $\mathcal{G}_i'|_{\strat_i'}$. Then, we have 
$\StairSc_e(w) = \StairSc_e(w(0)\cdots w(r-1)) +1 = \Sc_e(u') +1 = \Sc_e(w')$, 
and for $e' < e$, $\StairSc_{e'}(w) = \StairSc_{e'}(w(0)\cdots w(r-1)) = \Sc_{e'}(u') = \Sc_{e'}(w')$.
Finally, for $e'> e$, we have $\StairSc_{e'}(w) = 0 = \Sc_{e'}(w')$. 

In case of a push-transition~$\delta=(q,A, p, BC)$, we 
have $w(r) = (p, BC\gamma)$.
Consider the finite path 
\begin{equation*}
u''  = \WaluPush[P, c, p, BC] \rightarrow \WaluClaim[P, c, p, BC, R] \rightarrow \WaluCheck[p, B, R, \col(p), \col(p)]
\end{equation*}
in $\mathcal G_i'|_{\strat_i'}$ where $R$ is the prediction picked by $\strat_i'$. 
Notice that there is indeed an edge from $\last(u')$ to $\WaluPush[P, c, p, BC]$ in $E'|_{\strat_i'}$. 
We claim that $w' = u'u''$ has the desired properties. Requirement (\ref{req_lastvert})
is satisfied, as $\lastStrictBump(w)=w(r)$ in this case, and $\MinCol(w(r))=\col(p)$. Furthermore, 
$\WaluClaim[P, c, p, BC, R]$ is pushed onto the stack of $\mathcal T_\strat$ 
when processing~$\delta$. Hence, requirement (\ref{req_stack}) is satisfied.
  

The scores evolve as in the case of a skip-transition explained above, since 
in both cases we have $\lastBump(w) = w(r)$, and $u''$ contains exactly 
one vertex with color in $[n]$, namely its last vertex, which has the same color 
as $w(r)$. The intermediate auxiliary vertices have color~$n$ and therefore do not influence the 
scores we are interested in.

Finally, the case of a pop-transition is the most involved one, since a play in $\mathcal{G}_i'|_{\strat_i'}$ 
ends in a sink vertex, as soon as a pop-transition is simulated. In this case, 
$\mathcal{T}_\sigma$ uses the top $\WaluClaim$-vertex stored on its stack to
determine the appropriate $\WaluCheck$-vertex for being able to continue playing 
according to $\strat_i'$. 
Suppose the transition is $\delta = (q,A, p, \epsilon)$, i.e., we have 
$w(r) = (p, \gamma)$. 
Let $\delta_s=(q_s,A_s,q',BC)$ be the push-transition (of the PDS underlying $\mathcal G$) which induces
the edge $(w(s),w(s+1))\in E$. Note that $C\gamma_s=\gamma$, since the stack content $C\gamma_s$ 
remains untouched until $\delta$ is executed from $w(r-1)$ to $w(r)$. Hence, $w(r)=(p,C\gamma_s)$.
By definition of $\sigma$, there is an edge from $\last(u')=\WaluCheck[q, A, P, c, d]$ to 
$\WaluWin_i[p]$ in $E'|_{\strat_i'}$ such that $p \in P_c$.

Now, consider the run of $\mathcal T_\strat$ on $w$. By construction, 
the transducer pops the top $\WaluClaim$-vertex~$v$ from its stack while 
processing the transition $\delta$. We show that $v = \WaluClaim[P_s, c_s, q', BC, P]$. First, notice that $v$ was pushed onto the stack
while processing the transition from $w(s)$ to $w(s+1)$ which is induced by $\delta_s$. Applying the induction hypothesis shows that the run of 
$\mathcal T_\strat$ on the sequence of transitions induced by $w(0)\cdots w(s)$ ends in state $\last(u_s')=\WaluCheck[q_s,A_s,P_s,c_s,d_s]$ with some stack content
$\gamma_\strat \in(\Gamma^\strat)^+\bot$ satisfying the above requirements. 
Since now $\delta_s$ is to be processed, the run of $\mathcal T_\strat$ is continued as follows for some $R\in\Prediction$:
\begin{align*}
(\last(u_s'),\gamma_\strat) & \overset{\delta_s}{\mapstochar\relbar}(\WaluPush[P_s, c_s, q', BC],\gamma_\strat) 
\overset{\epsilon}{\mapstochar\relbar}(\WaluClaim[P_s, c_s, q', BC, R],\gamma_\strat)\\
& \overset{\epsilon}{\mapstochar\relbar}(\WaluCheck[q', B, R, \col(q'), \col(q')], \WaluClaim[P_s, c_s, q', BC, R]\cdot\gamma_\strat )
\end{align*}
It remains to show that $R=P$, which is done by applying the induction hypothesis to the run of 
$\mathcal T_\strat$ on transitions induced by $w(0)\cdots w(r-1)$. The top symbol $\WaluClaim[P_s, c_s, q', BC, R]$, which is pushed on the stack
while processing $(w(s),w(s+1))$, remains untouched until $w(r-1)$ is reached and is again the top symbol after processing $(w(r-2),w(r-1))$. 
However, since $\last(u')=\WaluCheck[q,A,P,c,d]$ is the state reached by $\mathcal T_\strat$ after processing $w(0)\cdots w(r-1)$ it follows from 
requirement (\ref{req_stack}) that $R=P$.

Consider the following finite path in $\mathcal G_i'|_{\strat_i'}$:
\[u''= \WaluPush[P_s, c_s, q', BC] \rightarrow v \rightarrow \WaluJump[p,C,P_s,c_s,c] \rightarrow \WaluCheck[p,C,P_s,\Min\{c_s,c,\col(p)\},\Min\{c,\col(p)\}].\]
Notice that there is an edge from $\last(u_s')$ to $\WaluPush[P_s, c_s, q', BC]$ in $E'|_{\strat_i'}$. 
So, we can show that $w'=u_s'u''$ satisfies the above requirements.
Requirement (\ref{req_lastvert}) is satisfied, since
\begin{align*}
& \,\MinCol(\lastStrictBump(w))\\
= & \,\Min\{\MinCol(\lastStrictBump(w(0)\cdots w(s))),\MinCol(w(s+1)\cdots w(r-1)),\col(w(r))\}\\
= & \,\Min\{c_s,\MinCol(\lastStrictBump(w(0)\cdots w(r-1))),\col(p)\}\\
= & \,\Min\{c_s,c,\col(p)\}\enspace.
\end{align*}
Requirement (\ref{req_stack}) is satisfied, since after processing $\delta$ by $\mathcal T_\strat$, the top stack symbol $v$ is popped
from the stack and the state $\WaluCheck[p,C,P_s,\Min\{c_s,c,\col(p)\},\Min\{c,\col(p)\}]$ is reached. By doing so, the same stack
content is reestablished as after the run of $\mathcal T_\strat$ on $\Reset(w)$. Hence, by applying the induction hypothesis, we have $C_0\cdots C_k=\gamma_s$.
Since we have $\gamma=C\gamma_s$, this suffices.
To show requirement (\ref{req_score}), let 
\begin{align*}
e & =\MinCol(\lastBump(w))\\
  & =\Min\{\MinCol(\lastStrictBump(w(0)\cdots w(r-1))),\col(w(r))\}\\
  & =\Min\{c,\col(p)\}\enspace. 
\end{align*}
Notice that $e$ is also the color of  $\last(w')=\WaluCheck[p,C,R,\Min\{c_s,c,\col(p)\},\Min\{c,\col(p)\}]$ in $\mathcal{G}_i'|_{\strat_i'}$. Thus, 
$\StairSc_e(w) = \StairSc_e(w(0)\cdots w(s)) +1 = \Sc_e(u_s') +1 = \Sc_e(w')$
and for $e' < e$
$\StairSc_{e'}(w) = \StairSc_{e'}(w(0)\cdots w(s)) = \Sc_{e'}(u_s') = \Sc_{e'}(w')$.
Finally, if $e'> e$, $\StairSc_{e'}(w) = 0 = \Sc_{e'}(w')$. 
\end{proof}

Now, the proof of Theorem \ref{thm_main} is straightforward. 

\begin{proof}[Proof of Theorem \ref{thm_main}]
Assume that Player $i$ wins $\mathcal G$, then he also wins $\mathcal G_i'$ due to Theorem~\ref{thm_walu}. 
For every color $c\in[n]$, there are at most 
$k_\mathcal{G}$ $\WaluCheck$-vertices colored by $c$. Hence, due to Remark~\ref{remark} 
there is a positional winning 
strategy $\strat_i'$ in $\mathcal G_i'$ for Player~$i$ such 
that for every $c\in[n]$ with $\Parity(c)=1-i$, $\MaxSc_c(\rho')\leq k_\mathcal{G}$, 
for every play $\rho'$ which is consistent with 
$\strat_i'$. 
From Lemma \ref{walu_scores_are_equivalent_to_stair_scores}, it follows that
the pushdown strategy $\strat$ which is 
constructed from $\strat_i'$ bounds the stair-scores 
of Player $1-i$ by $k_\mathcal{G}$. Thus, for every
play $\rho$ which is consistent with $\strat$ and every $k>k_\mathcal{G}$, 
there exists $w\sqsubset \rho$ such that $w$ is winning for Player~$i$ in $\mathcal G_k$.
Hence, using the same strategy $\strat$ Player $i$ wins every finite-time game $\mathcal G_k$ for $k>k_\mathcal{G}$.
The other direction follows by determinacy of parity games.
\end{proof}
\section{Lower Bounds}
\label{sec_lowerbounds}

In the previous section, we proved the equivalence between pushdown games and
corresponding finite-time pushdown games with an exponential threshold. In this
section, we present an (almost) matching lower bound on the threshold that
always yields equivalent games. To this end, we construct a pushdown game in
which the winning player is forced to reach a configuration of \emph{high} stack
height while only visiting states colored by a bad color for him. Thereby, the
opponent is the first player to reach \emph{high} stair-scores, although he loses the play
eventually.

\begin{theorem}
There are a family of pushdown games~$(G_n,\col_n)$ and thresholds~$k_n$
exponential in the cubic root of the size of the underlying PDS such that for
every $n>0$, Player~$0$ wins the pushdown game~$(G_n,\col_n)$, but for every
$k\leq k_n$, Player~$1$ wins the finite-time pushdown game~$(G_n,\col_n, k)$.
\end{theorem}

\begin{proof}
We denote the $i$-th prime number by $p_i$. For $n>0$, let $k_n = \prod_{i=1}^n p_i$ and define the PDS $\mathcal P_n=(Q_n,\{A\},\Delta_n, q_\mathrm{in})$ as follows:
$Q_n = \{q_\mathrm{in}, q_\Box\} \cup \bigcup_{i=1}^{n} M_i$,
where $M_i=\{q_i^j \mid 0 \leq j < p_i\}$,
and $\Delta_n$ consists of the following transitions:
\begin{itemize}
 \item $(q_\mathrm{in}, X, q_\mathrm{in}, AX)$ and $(q_\mathrm{in}, X, q_\Box,
AX)$ for every $X\in\{A,\bot\}$,
 \item $(q_\Box, A, q_i^0, A)$ for every $1 \leq i \leq n$,
 \item $(q_i^j, A, q_i^{\ell}, \varepsilon)$, where $\ell = (j+1) \bmod p_i$, and
 \item $(q, \bot, q, \bot)$, for every $q\in Q_n\setminus\{q_\mathrm{in}\}$.
\end{itemize}
To specify the partition of $Q_n$, let $q_\Box$ belong to Player~$1$. All other
states are Player~$0$ states. The coloring is given by
$\col_n(q_i^0)= 0$ for every $1\leq i \leq n$ and $\col_n(q) =1$ for every
other state~$q$. We have $k_n \ge 2^n$ and $|Q_n|$ can be bounded from above by $\mathcal{O}(n^2\log(n))$. Hence,
$k_n$ is exponential in the cubic root of $|Q_n|$.
The pushdown game $(G_2,\col_2)$ is depicted in
Figure~\ref{fig_pushdown_game_G_2}. Double-lined vertices are those colored by~$0$. 

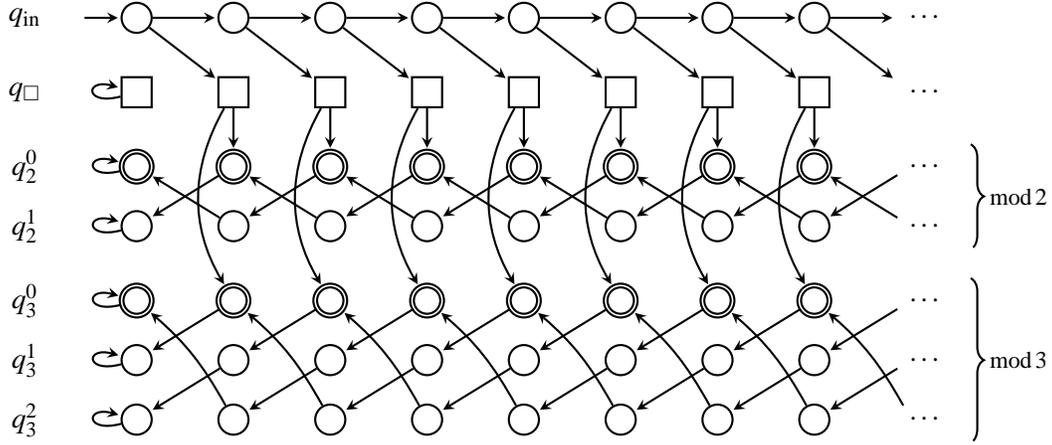
\begin{figure}
\begin{center}
\begin{tikzpicture}[scale=.99]

\foreach \x in {0,...,7}
  \node[p0] at (1.3*\x,0) (a\x) {};

\node at (10.6,0) (aend) {$\cdots$};
\node at (10.6,-1) (bend) {$\cdots$};
\node at (10.6,-2) (mend) {$\cdots$};
\node at (10.6,-2.8) (mmend) {$\cdots$};
\node at (10.6,-3.8) (mmmmend) {$\cdots$};
\node at (10.6,-4.6) (mmmmmend) {$\cdots$};
\node at (10.6,-5.4) (mmmmmmend) {$\cdots$};

\foreach \x in {0,...,7}
  \node[p1] at (1.3*\x,-1) (b\x) {};

\node[p0, double] at (0,-2) (m200) {};

\foreach \x in {1,...,7}
    \node[p0, double] at (1.3*\x,-2) (m20\x) {};

\foreach \x in {0,...,7}
    \node[p0] at (1.3*\x,-2.8) (m21\x) {};


\node[p0, double] at (0,-3.8) (m300) {};

\foreach \x in {1,...,7}
    \node[p0, double] at (1.3*\x,-3.8) (m30\x) {};

\foreach \x in {0,...,7}
    \node[p0] at (1.3*\x,-4.6) (m31\x) {};

\foreach \x in {0,...,7}
    \node[p0] at (1.3*\x,-5.4) (m32\x) {};


\node[p0, draw=none] at (1.3*8,0) (a8) {};
\node[p1, draw=none] at (1.3*8,-1) (b8) {}; 
\node[p0, draw=none] at (1.3*8,-2) (m208) {}; 
\node[p0, draw=none] at (1.3*8,-2.8) (m218) {}; 
\node[p0, draw=none] at (1.3*8,-3.8) (m308) {}; 
\node[p0, draw=none] at (1.3*8,-4.6) (m318) {}; 
\node[p0, draw=none] at (1.3*8,-5.4) (m328) {};

\node at (-1.5, 0) (alabel) {$q_\mathrm{in}$};
\node at (-1.5,-1) (blabel) {$q_\Box$};
\node at (-1.5,-2) (clabel) {$q_2^0$};
\node at (-1.5,-2.8) (alabel) {$q_2^1$};
\node at (-1.5,-3.8) (blabel) {$q_3^0$};
\node at (-1.5,-4.6) (clabel) {$q_3^1$};
\node at (-1.5,-5.4) (clabel) {$q_3^2$};

\path[draw, thick]
(a0) edge (a1)
(a1) edge (a2)
(a2) edge (a3)
(a3) edge (a4)
(a4) edge (a5)
(a5) edge (a6)
(a6) edge (a7)
(a7) edge (a8);

\path[draw, thick]
(a0) edge (b1)
(a1) edge (b2)
(a2) edge (b3)
(a3) edge (b4)
(a4) edge (b5)
(a5) edge (b6)
(a6) edge (b7)
(a7) edge (b8);

\path[draw, thick]
(b1) edge[bend left=0] (m201)
(b1) edge[bend right=30] (m301)
(b2) edge[bend left=0] (m202)
(b2) edge[bend right=30] (m302)
(b3) edge[bend left=0] (m203)
(b3) edge[bend right=30] (m303)
(b4) edge[bend left=0] (m204)
(b4) edge[bend right=30] (m304)
(b5) edge[bend left=0] (m205)
(b5) edge[bend right=30] (m305)
(b6) edge[bend left=0] (m206)
(b6) edge[bend right=30] (m306)
(b7) edge[bend left=0] (m207)
(b7) edge[bend right=30] (m307);

\path[draw, thick]
(m201) edge (m210)
(m202) edge (m211)
(m203) edge (m212)
(m204) edge (m213)
(m205) edge (m214)
(m206) edge (m215)
(m207) edge (m216)
(m208) edge (m217);

\path[draw, thick]
(m211) edge (m200)
(m212) edge (m201)
(m213) edge (m202)
(m214) edge (m203)
(m215) edge (m204)
(m216) edge (m205)
(m217) edge (m206)
(m218) edge (m207);


\path[draw, thick]
(m301) edge (m310)
(m302) edge (m311)
(m303) edge (m312)
(m304) edge (m313)
(m305) edge (m314)
(m306) edge (m315)
(m307) edge (m316)
(m308) edge (m317);

\path[draw, thick]
(m311) edge (m320)
(m312) edge (m321)
(m313) edge (m322)
(m314) edge (m323)
(m315) edge (m324)
(m316) edge (m325)
(m317) edge (m326)
(m318) edge (m327);

\path[draw, thick]
(m321) edge[bend right = 10] (m300)
(m322) edge[bend right = 10] (m301)
(m323) edge[bend right = 10] (m302)
(m324) edge[bend right = 10] (m303)
(m325) edge[bend right = 10] (m304)
(m326) edge[bend right = 10] (m305)
(m327) edge[bend right = 10] (m306)
(m328) edge[bend right = 10] (m307);


\path[draw, thick]
(b0) edge[loop left] ()
(m200) edge[loop left] ()
(m210) edge[loop left] ()
(m300) edge[loop left] ()
(m310) edge[loop left] ()
(m320) edge[loop left] ();

\path[draw, thick]
(-.7,0) edge (a0);

\draw[decorate, thick, decoration={brace,amplitude=5}] (11.2,-1.7) -- node[auto, xshift =2]{\footnotesize$\bmod 2$} (11.2,-3.1);

\draw[decorate, thick, decoration={brace,amplitude=5}] (11.2,-3.5) -- node[auto, xshift =2]{\footnotesize$\bmod 3$} (11.2,-5.7);

\end{tikzpicture}
\end{center}
\caption{Pushdown Game $(G_2, \col_2)$}
\label{fig_pushdown_game_G_2}
\end{figure}

A play in the game $(G_n,\col_n)$ proceeds as follows. Player~$0$ picks a
natural number $x > 0$ by moving the token to the configuration $(q_\Box,
A^{x}\bot)$. If he fails to do so by staying in state $q_\mathrm{in}$ ad
infinitum he loses, since $\col_n(q_\mathrm{in})=1$. At $(q_\Box, A^{x}\bot)$,
Player~$1$ picks a modulus $p_i \in\{p_1,\ldots,p_n\}$ by moving the token to
$(q^0_i, A^x\bot)$. From this configuration, a single path emanates, i.e., there
is only one way to continue the play. Player~$0$ wins this play if and only if $x \bmod p_i
= 0$. Hence, Player~$0$ has a winning strategy for this game by moving the token
to some non-zero multiple of $k_n$, i.e., Player~$0$ wins $(G_n,\col_n)$.

Now, let $k \leq k_n$. If Player~$0$ reaches $(q_\mathrm{in}, A^{k-1}\bot)$,
then he loses the finite-time pushdown game $(G_n,\col_n,k)$, since in this
case Player~$1$ reaches stair-score $k$ for color~$1$. On the other hand, if he
moves the token to a configuration $(q_\Box, A^x\bot)$ for some $x \le k-1$,
then there is a $p_i \in \{p_1,\ldots,p_n\}$ such that $x \bmod p_i \neq 0$, as $x <
k_n$. Hence, assume Player~$1$ moves the token to $(q_i^0, A^x\bot)$. Then, the
play ends in a self-loop at a configuration~$(q_i^m, \bot)$ for some $m\neq 0$.
The path $w$ from $(q_\mathrm{in}, \bot)$ to $(q_i^m,\bot)$ via $(q_\Box, A^x)$
satisfies $\MaxStairSc_0(w) \le x$. Since $q_i^m$ is colored by $1$, the scores
of Player~$0$ are never increased while using the self-loop at $(q_i^m,\bot)$.
Thus, his scores never reach the threshold $k$. Hence, Player~$1$ is the first
to reach this threshold, since Lemma~\ref{every_k_is_reached_pushdownversion}
guarantees that there is some color that reaches the threshold eventually. Thus,
Player~$1$ wins $(G_n, col_n, k)$.
\end{proof}

\section{Conclusion}

We have shown how to play parity games on pushdown graphs in finite time. To
this end, we adapted the notions of scoring functions to exploit the intrinsic
structure of a pushdown game graph to obtain an finite-duration game that always has
the same winner as the infinite game. Thus, the winner of a parity game on a
pushdown game graph can be determined by solving a finite reachability game.

This work transfers results obtained for games on finite game graphs to infinite 
graphs. In ongoing work, we investigate if and
how a winning strategy for the safety game, in which Player~$0$ wins if
and only if he prevents his opponent from reaching an exponential stair-score
can be turned into a winning strategy for the original pushdown
game. The winner of these two games is equal, due to
Lemma~\ref{walu_scores_are_equivalent_to_stair_scores}. 

On the other hand, our
results could be extended by considering more general classes of infinite graphs
having an intrinsic structure, e.g., configuration graphs of higher-order
pushdown systems. Finally, there is a small gap between the upper and lower
bound on the threshold score that always yields an equivalent finite-duration
pushdown game, which remains to be closed.

\bibliographystyle{eptcs}
\bibliography{references}
\end{document}